\documentclass[11pt,a4paper]{article}

\usepackage[a4paper, total={5.8in, 9in}]{geometry}

\usepackage{amsmath,amsfonts}
\usepackage{amsthm}
\usepackage{amssymb,hyperref}
\usepackage[english]{babel}
\usepackage{enumitem}
\usepackage{mdframed}

\usepackage{epsfig}
\usepackage{bbm}
\usepackage{thm-restate}

\usepackage{xcolor} 
\hypersetup{
    colorlinks,
    linkcolor={blue},
    citecolor={blue},
    urlcolor={blue}
}

\newtheorem{theorem}{Theorem}[section]
\newtheorem{lemma}[theorem]{Lemma}
\newtheorem{conjecture}[theorem]{Conjecture}

\newtheorem{observation}[theorem]{Observation}

\newtheorem{claim}[theorem]{Claim}
\newtheorem{definition}[theorem]{Definition}
\numberwithin{theorem}{section}

\newenvironment{cproof}
{\begin{proof}
 [Proof.]
 \vspace{-1.5\parsep}
}
{ \end{proof}}

\newcommand{\CE}{C_k^{H}}

\newcommand{\CC}{C_{2n+1}}
\newcommand{\KE}{K^{C_{2n+1}}_{Even}}

\newcommand{\KKK}{K_3^{C_{2n+1}}}

\newcommand{\BT}{B_{2n+1}}

\def\pab{{\vec{p}_{ab}}}
\def\vc{{\vec{C}_{2n+1}}}
\def\vcs{{\vec{C}_{2n+1}}}

\def\lf{{$\ell_f$}}
\def\lfs{{$\ell_f$}}

\begin{document} \bibliographystyle{alpha}
\title{Explicit 3-colorings for exponential graphs}

\author{\textsc{Adrien Argento\thanks{Universit\'e Grenoble Alpes.
      Supported by a Stage d'Excellence from UGA.}} \and {\sc Pierre
  Charbit\thanks{Universit\'e Paris Diderot.}} \and \textsc{Alantha
    Newman\thanks{Universit\'e Grenoble Alpes.  Supported in
      part by IDEX-IRS SACRE.}}}

\date{March 13, 2019}
\maketitle

\begin{abstract}
For a graph $H$ and integer $k \geq 1$, two functions $f, g$ from
$V(H)$ into $\{1, \dots, k\}$ are adjacent if for all edges $uv$ of
$H$, $f(u) \neq g(v)$. The graph of all such functions is the
exponential graph $K_k^H$.  El-Zahar and Sauer proved that if $\chi(H)
\geq 4$, then $K_3^H$ is 3-chromatic~\cite{el1985chromatic}.  Tardif
showed that, implicit in their proof, is an algorithm for 3-coloring
$K_3^H$ whose time complexity is polynomial in the size of
$K_3^H$~\cite{tardifAlg}.  Tardif then asked if there is an
``explicit'' algorithm for finding such a coloring: Essentially, given
a function $f$ belonging to a 3-chromatic component of $K_3^H$, can we
assign a color to this vertex in time polynomial in the size of $H$?
The main result of this paper is to present such an algorithm,
answering Tardif's question affirmatively.  Our algorithm yields an
alternative proof of the theorem of El-Zahar and Sauer that the
categorical product of two 4-chromatic graphs is 4-chromatic.
\end{abstract}

\section{Introduction}

For a graph $G$, we use $V(G)$ and $E(G)$ to denote its vertex and
edge sets, respectively, and we use $\chi(G)$ to denote its chromatic
number.  A homomorphism from a graph $G$ to a graph $G'$ is a function
$\phi$ from $V(G)$ to $V(G')$ such that for every edge $uv$ in $G$,
$\phi(u)\phi(v)$ is an edge in $G'$.  We denote by $G \rightarrow G'$
the existence of a homomorphism from $G$ to $G'$.  Note that a graph
$G$ admits a proper $k$-coloring if and only if $G \rightarrow K_k$.

The categorical product of two graphs $G \times H$ has vertex set
$V(G) \times V(H)$ and edge set $((u,v),(\bar{u},\bar{v}))$ for
$u\bar{u}$ and $v\bar{v}$ belonging to $E(G)$ and $E(H)$,
respectively.  Observe that $G \times H$ admits a homomorphism to both
$G$ and $H$.  Since a proper colouring corresponds to a homomorphism
to a complete graph, and since $G\times H\rightarrow G$ and $G\times
H\rightarrow H$, it is therefore immediate that $\chi(G \times H) \leq
\min\big(\chi(G), \chi(H)\big)$.  The following conjecture is due to
Hedetniemi~\cite{hedetniemi1966homomorphisms} and was also posed as a
question by Greenwell and Lov{\'a}sz~\cite{greenwell1974applications}.
\begin{conjecture}\label{conj:first}
$\chi(G \times H) = \min\big(\chi(G), \chi(H)\big)$.
\end{conjecture}

For two graphs $H$ and $K$, there exists an {\em exponential graph}
$K^{H}$ with the following property: $G\times H \rightarrow K$ only if
$G\rightarrow K^{H}$.  The vertices of $K^H$ are functions from $V(H)$
into $V(K)$, and two functions $f, g$ are adjacent if for every edge
$uv$ of $H$, $f(u)g(v)$ is an edge of $K$.  With this definition of an
exponential graph, observe that Hedetniemi's conjecture can be
rewritten: If $\chi(H)>k$ and $G\rightarrow K_k^{H}$, then
$G\rightarrow K_{k}$. But now we see that it is sufficient to replace
$G$ with $K_k^{H}$, and Conjecture \ref{conj:first} can be restated as
follows.
\begin{conjecture}\label{conj:second}
If $\chi(H)>k$, then $\chi(K_k^{H})\leq k$.
\end{conjecture}

The connection between exponential graphs and Hedetniemi's conjecture
was observed by El-Zahar and Sauer who proved Conjecture
\ref{conj:second} when $k=3$~\cite{el1985chromatic}.  
(Exponential
graphs have also been studied in other contexts~\cite{lovasz1967operations}.)  
Specifically, El-Zahar and Sauer
proved that if $\chi(H) > 3$, then $K_{3}^{H}$ is
$3$-colorable.  Their proof is based on a global parity argument
concerning so-called fixed points of odd cycles and does not
immediately yield a 3-coloring of $K_3^H$.  Attributing the question
of efficiently finding a 3-coloring to Edmonds, Tardif presented an
algorithm, implicit in the work of El-Zahar and Sauer, for 3-coloring
$K_3^H$~\cite{tardifAlg}.  Basically, Tardif noted that if we consider
an (arbitrary) edge $ab$ in any odd cycle of $H$, then functions $f$
in which $f(a) = f(b)$ form a hitting set of the odd cycles in
$K_3^H$.  A 3-coloring can easily be found based on a bipartition of
the remaining vertices in $K_3^H$.  The time complexity of this
algorithm depends on the time to find a bipartition, which is
polynomial in the size of $K_3^H$, but can be exponential in $|V(H)|$.
A thorough description of this algorithm is provided in Section
\ref{sec:tardif-alg}.

Tardif then posed the problem of finding an ``explicit'' 3-coloring of
$K_3^H$.  Essentially, given a function $f$ belonging to a 3-chromatic
component of $K_3^H$, can we assign a color to this
vertex in time polynomial in the size of $H$?  (His precise question
is a bit more involved and presented in detail in Section
\ref{sec:tardif-alg}.)  In this paper, we present an
algorithm whose time complexity is linear in $|V(H)|$ for finding such an
explicit 3-coloring.

\subsection{From coloring $K_3^{C_{2n+1}}$ to coloring
  $K_3^H$}\label{from-coloring}

Let $n$ be a positive integer, and let $C_{2n+1}$ denote the odd cycle
on the vertices $u_1, u_2, \dots, u_{2n+1}$.  The vertex set $V(\KKK)$
consists of all $3^{2n+1}$ functions from $V(C_{2n+1})$ into
$\{1,2,3\}$.  A vertex $u_i \in C_{2n+1}$ is a {\em fixed point} in $f
\in V(\KKK)$ if $f(u_{i-1}) \neq f(u_{i+1})$, where indices are
computed modulo $2n+1$.  Let $\KE$ denote the subgraph of $\KKK$
induced on the following vertex set.
\begin{eqnarray*}
V(\KE) & = & \left\{f \in V(\KKK):~ f \text{ has an even number
  of fixed points}\right\}.
\end{eqnarray*}

The problem of finding a 3-coloring of $K_3^H$ when $\chi(H) > 3$ can
be reduced to the problem of finding a 3-coloring for
$\KE$~\cite{el1985chromatic, tardifAlg}.  Any non-isolated function
$f$ from $V(H)$ into $\{1,2,3\}$ contains an odd cycle with an even
number of fixed points and such an odd cycle can be found in time
polynomial in the size of $H$ (see Proposition 4.1 in
\cite{el1985chromatic} or Claim 2 from
\cite{zhu1998survey}).\footnote{Observe that the proof of this claim
  becomes easier when $\chi(H) \geq 5$.  Consider any function $f$
  from $V(H)$ into $\{1,2,3\}$ and partition $V(H)$ into the set of
  vertices colored by $1$ and $2$ and the set of vertices colored by
  $3$.  Either the first set or the second set contains an odd cycle
  (which then has an even number of fixed points because it has at
  most two colors) or each set is bipartite, and we can color with
  four colors, which is a contradiction.}  For the connected component
of $K_3^H$ containing this function $f$, let us fix this odd cycle in
$H$ to be $C_{2n+1}$.  Now we find a 3-coloring for $\KE$.  Each
function $g$ in the same connected component of $K_3^H$ as $f$ must
also have an even number of fixed points on $C_{2n+1}$ (see Lemma 3.3
in \cite{el1985chromatic} or Claim 1 from \cite{zhu1998survey}).  So
for each function $g \in V(K_3^H)$ in the same connected component as
$f$, the restriction of $g$ onto $C_{2n+1}$ belongs to $V(\KE)$.
Therefore, the vertex in $V(K_3^H)$ corresponding to the function $g$
can be assigned the same color that the restriction of $g$ onto
$C_{2n+1}$ receives in the 3-coloring of $\KE$.

\subsection{Algorithm for 3-coloring $\KE$}\label{sec:tardif-alg}

Now we are ready to present the algorithm for 3-coloring $\KE$ from
\cite{tardifAlg}.  Let $ab$ denote a fixed (but arbitrarily chosen)
edge in $C_{2n+1}$.

\vspace{5mm}

\noindent
\fbox{\parbox{16cm}{

{\sc Color-Graph($\KE$)}

\begin{enumerate}

\item For all $f \in V(\KE)$:

{\begin{enumerate}
\item[(i)] If $f(a) = f(b)$, assign the vertex $f$
  in $V(\KE)$ the color $f(a)$.
\end{enumerate}
}

\item Remove all the colored vertices from $V(\KE)$.

\item Find a bipartition $(A,B)$ of the subgraph induced on the
  remaining vertices.

\item For each vertex $f$ in $A$,
  assign the vertex color $f(a)$.  

\item For each vertex $f$ in $B$,
  assign the vertex color $f(b)$.  

\end{enumerate}

}}

\vspace{3mm} The correctness of this algorithm follows from the fact
that the copies of $C_{2n+1}$ in which $f(a) = f(b)$ form a hitting
set for the odd cycles in $\KE$, which follows from the main result of
El-Zahar and Sauer (e.g., see Lemma 3.1 and Proposition 3.4 in
\cite{el1985chromatic}).  

\subsection{Tardif's open problem}

Tardif asked 
if there is an
algorithm, whose running time is polynomial in $n$, to assign a color
to $f \in V(\KE)$ so that a 3-coloring is maintained for any subset of
colored vertices of $\KE$.  (See Problem 6 in \cite{tardifAlg} and also
\cite{tardifFields}.)  He defines the vertex set
\begin{eqnarray*}
V(\BT) & = & \left\{f \in V(\KE):~ f(a) \neq f(b)\right\},
\end{eqnarray*}
for a fixed (but arbitrarily chosen) edge $ab \in C_{2n+1}$.  $\BT$ is
a bipartite subgraph of $\KE$ induced on $V(\BT)$.  If we can decide
in $O(n)$ time to which side of the bipartition $f \in V(\BT)$
belongs, then we can resolve Tardif's question affirmatively.  We
present an algorithm for this task in the next section.  Our approach
is inspired by ideas from reconfiguration of
3-recolorings~\cite{cereceda2007finding}.

We note that Tardif showed that the main result from
\cite{el1985chromatic} implies that $\BT$ is bipartite.  Conversely,
our algorithm gives another proof that $\BT$ is bipartite, and
consequently, we give an alternative proof of the main result of
\cite{el1985chromatic}.  Note, however, that our proof is not
completely independent as it uses Lemma 3.3 and Proposition 4.1 from
\cite{el1985chromatic}.

\section{Properties of adjacent functions}

In this section, we state and prove two properties of functions that
are adjacent in $\KKK$.  These properties are key to the design and
analysis of our algorithm, which we present in Section
\ref{coloring-alg}.  For an ordered pair of vertices $uv$ (i.e., an
    {\em arc} $uv$) with colors $c(u)$ and $c(v)$, respectively, we
    say the {\em value} of $uv$ is $\Delta(c(u)c(v))$, where
\begin{eqnarray*}
\Delta(12) = \Delta(23) = \Delta(31) = +1 ~\text{ and }~
\Delta(21) = \Delta(32) = \Delta(13) = -1. 
\end{eqnarray*}
Monochromatic pairs have value 0.

Let $f$ denote a function from $V(\CC)$ into $\{1,2,3\}$.  For a
vertex $u \in \CC$, its color in $f$ is denoted by $f(u)$.  We fix the
orientation for the chords of length two in the cycle $\CC$ so that
they form the directed cycle $\{u_1, u_3, \dots u_{2n+1}, u_2, u_4,
\dots, u_{2n}, u_1\}$, which we refer to as $\vc$.  Then we have the
following definitions. Recall that $ab$ is a fixed (but arbitrarily
chosen) edge in $\CC$.
\begin{definition}\label{def:chord-cycle}
The {\em label} of $f$, denoted by \lf, is the total value of the arcs
in $\vcs$ based on $f$.  Formally,
$$\ell_f = \sum_{uv  \in \vc} \Delta(f(u)f(v)).$$
\end{definition}

\begin{definition}\label{def:little-path}
The {\em little path} of $f$, denoted by $\pab$, is the directed path
from $a$ to $b$ in $\vcs$ containing $n$ arcs (e.g., see Figure
\ref{fig:vec}).  The value of $\pab$, denoted by $p_f$, is computed as
follows.
$$p_f = \sum_{uv  \in \pab} \Delta(f(u)f(v)).$$
\end{definition}

\begin{figure}[h]
\begin{center}
\epsfig{file = 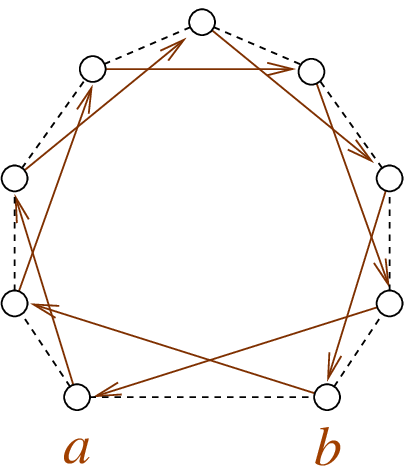, width=3cm} \hspace{15mm}
\epsfig{file = 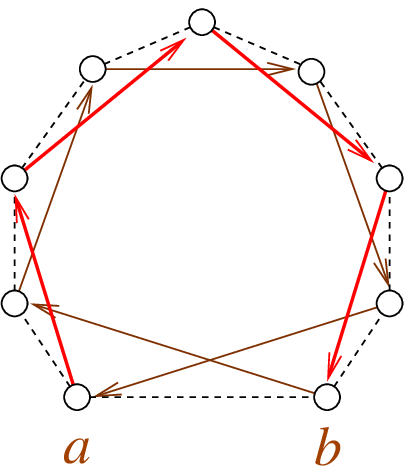, width=3cm} 
\caption{The dotted edges denote $C_9$.  The directed edges denote
  $\vec{C}_9$. The little path $\pab$ is shown in red.}
\label{fig:vec}
\end{center}
\end{figure}

\begin{observation}\label{zero-mod-3}
\lfs $\equiv 0 \pmod 3$.
\end{observation}

\begin{observation}\label{notmult}
If $f(a) \neq f(b)$, then $p_f \not \equiv 0 \pmod 3$.
\end{observation}

Now let $f$ and $g$ be two functions from $V(C_{2n+1})$ into
$\{1,2,3\}$ such that $f$ and $g$ are adjacent.  Recall that $\{u_1,
u_2, \dots, u_{2n+1}\}$ denotes the vertex set $V(\CC)$.  Let $f_i$
denote the copy of $u_i$ in $f$ and let $g_i$ denote the copy of $u_i$
in $g$. Define the directed cycle $D$ as follows.
\begin{eqnarray}
D & = & \{f_1, g_2, f_3, \dots, g_{2n}, f_{2n+1}, g_1, f_2, \dots,
f_{2n}, g_{2n+1}, f_1\}.\label{directed-cycle}
\end{eqnarray}
We have
\begin{eqnarray}
\Delta(D) = \sum_{i=1}^{2n+1} \Delta(f(u_i)g(u_{i+1})) +
\sum_{i=1}^{2n+1} \Delta(g(u_i)f(u_{i+1})),\label{delta-cycle}
\end{eqnarray} 
where subscripts are computed modulo $2n+1$.  We can relate the value
of the arcs in $\vc$ to the value of the arcs in $D$ using the
following claim.
\begin{claim}\label{clm:map}
Let $f$ and $g$ be two functions from $V(\CC)$ into $\{1,2,3\}$ such
that $f$ and $g$ are adjacent.  Then
$$\Delta(f(u_i)f(u_{i+2})) = -\frac{\Delta(f(u_i)g(u_{i+1}))
  + \Delta(g(u_{i+1})f(u_{i+2}))}{2}.$$
\end{claim}

\begin{cproof}
Since $f$ and $g$ are adjacent, we have $f(u_i) \neq g(u_{i+1})$ and
$f(u_i) \neq g(u_{i-1})$ for $i \in \{1, \dots, k\}$.  If $f(u_i) =
f(u_{i+2})$, then $\Delta(f(u_i)g(u_{i+1})) +
\Delta(g(u_{i+1})f(u_{i+2})) = 0$. Furthermore, we have the following
observations.

\begin{enumerate}

\item  $\Delta(f(u_i)f(u_{i+2})) = +1 \implies
\Delta(f(u_i)g(u_{i+1})) + \Delta(g(u_{i+1})f(u_{i+2})) = -2$, and

\item $\Delta(f(u_i)f(u_{i+2})) = -1 \implies \Delta(f(u_i)g(u_{i+1}))
  + \Delta(g(u_{i+1})f(u_{i+2})) = +2$.
\end{enumerate}\end{cproof}

Next, we present two key properties of adjacent functions via the
following lemmas.

\begin{lemma}\label{label-lemma}
Let $f$ and $g$ be two functions from $V(\CC)$ into $\{1,2,3\}$ such
that $f$ and $g$ are adjacent.
Then $\ell_{f} = \ell_{g}$.
\end{lemma}

\begin{proof}
Recall that $\ell_f$ is defined as follows.
$$\ell_{f} = \sum_{i=1}^{2n+1} \Delta(f(u_i)f(u_{i+2})).$$ By Claim
\ref{clm:map}, we observe that $\ell_{f} = -\Delta(D)/2$.  The same
argument shows that $\ell_{g} = -\Delta(D)/2$.\end{proof}

\begin{lemma}\label{little-path-lemma}
Let $f$ and $g$ be two functions from $V(C_{2n+1})$ into $\{1,2,3\}$
such that $f$ and $g$ are adjacent and let $\ell = \ell_f =\ell_g$.
Then $\ell -1 \leq p_f + p_g \leq \ell + 1$.
\end{lemma}

\begin{proof}
Without loss of generality, let $a = u_1$ and $b=u_{2n+1}$.  Then
$$p_f = \Delta(f(u_1)f(u_{3})) + \Delta(f(u_3)f(u_5)) 
+ \dots +
\Delta(f(u_{2n-1})f(u_{2n+1})),$$ and
$$p_g = \Delta(g(u_1)g(u_{3})) + \Delta(g(u_3)g(u_5)) 
+ \dots +
\Delta(g(u_{2n-1})g(u_{2n+1})).$$
Applying Claim \ref{clm:map}, we have
\begin{eqnarray*}
p_f + p_g & = & 
-\frac{\Delta(D) - \Delta(f(b)g(a)) -
  \Delta(g(b)f(a))}{2} \\
& = & \ell + \frac{\Delta(f(b)g(a)) +
  \Delta(g(b)f(a))}{2}.
\end{eqnarray*}
Since
\begin{eqnarray*}
-2 & \leq & \Delta(f(b)g(a)) +
  \Delta(g(b)f(a)) \quad \leq \quad 2, 
\end{eqnarray*}
the lemma follows.
\end{proof}

\section{Coloring a vertex of $\KE$ in time $O(n)$}\label{coloring-alg}

Now we present an algorithm to assign a color to a function $f$ when
$f \in V(\KE)$ (i.e., $f \in V(\KKK)$ and \lfs $\equiv 0 \pmod 2$).
Our algorithm will produce a 3-coloring for $\KE$ and the time
required to assign a color to $f$ is $O(n)$.  We recall that edge $ab
\in \CC$ is a fixed edge and is considered as input to the algorithm.

\vspace{5mm}

\noindent
\fbox{\parbox{8cm}{

{\sc Color-Vertex($f \in \KE$)}

\begin{enumerate}

\item If $f(a) = f(b)$, assign $f$ color $f(a)$.

\item {Otherwise, compute value $p_f$.

\begin{enumerate}
\item If $p_f < \ell_f/2$,
  assign color $f(a)$.  

\item If $p_f > \ell_f/2$,
  assign color $f(b)$.

\end{enumerate}
}

\end{enumerate}

}}

\vspace{5mm}

The correctness of the algorithm will be shown via the following
theorem, whose proof is based on the observation that if $f$ and $g$
are adjacent in $\BT$, then by Lemma \ref{little-path-lemma}, $p_f$
and $p_g$ are anticorrelated.  For example, if $\ell_f = 0$, then either $p_f$
is positive and $p_g$ is negative or vice versa.  Theorem
\ref{thm:main} implies that using {\sc Color-Vertex} to color the
vertices of $\KE$ results in a proper 3-coloring of $\KE$.

\begin{theorem}\label{thm:main}
Let $f$ and $g$ be two functions from $V(\KE)$ such that $f$ and $g$
are adjacent.  Then 
{\sc Color-Vertex} assigns $f$ and $g$ different colors.
\end{theorem}

First, we show via Claim \ref{clm:eitherOr} that Step 2. of {\sc Color-Vertex} is well-defined.
\begin{claim}\label{clm:eitherOr}
Let $f$ be a function in $V(\BT)$.
Then either $p_f < \ell_f/2$ or $p_f > \ell_f/2$.
\end{claim}

\begin{cproof}
By Observation \ref{zero-mod-3}, $\ell_f$ is a multiple of 3.  Since
$f \in V(\BT)$, it follows that $\ell_f$ is a multiple of 2.  Thus,
$\ell_f/2$ is an integer and is a multiple of 3.

By Observation \ref{notmult} and the fact that $f \in V(\BT)$, $p_f$
is not a multiple of 3.  Therefore, $p_f \neq \ell_f/2$.  Thus, we
have
\begin{eqnarray*}
p_f & \leq & \ell_f/2 \implies p_f \leq \ell_f/2 -1, \text{ and }\\
p_f & \geq & \ell_f/2 \implies p_f \geq \ell_f/2 + 1.
\end{eqnarray*}
\end{cproof}

Next, we show that two adjacent functions in $\BT$ will not be
assigned to the same side of the bipartition, and thus will not be
assigned the same color.
\begin{claim} \label{clm:anticorrelation}
Let $f$ and $g$ be two adjacent functions in $V(\BT)$
 and let $\ell = \ell_f = \ell_g$.  Then $p_f$ and $p_g$ cannot both
 be greater (or both smaller) than $\ell/2$.
\end{claim}

\begin{cproof}
By Lemma \ref{little-path-lemma}, we have
\begin{eqnarray*}
\ell-1 \leq p_f +  p_g \leq \ell + 1.
\end{eqnarray*}
Let us consider two cases.  The first case is when 
$p_f \leq \ell/2 - 1$.  Then we have
\begin{eqnarray*}
\ell - 1 - p_g & \leq & p_f ~ \leq ~ \ell/2 -1 ~ \Rightarrow\\
p_g & \geq & \ell/2.
\end{eqnarray*}
By Claim \ref{clm:eitherOr}, we conclude that $p_g > \ell/2$.
The second case is when $p_f \geq \ell/2  + 1$.  Then we have
\begin{eqnarray*}
\ell/2 + 1 & \leq & p_f ~ \leq ~ \ell + 1 - p_g ~ \Rightarrow\\
p_g & \leq & \ell/2.
\end{eqnarray*}
By Claim \ref{clm:eitherOr}, we conclude that $p_g < \ell/2$.
\end{cproof}

\section{Explicit homomorphisms for odd cycles}

For a graph $H$ and a cycle $C_k$ for odd integer $k \geq 5$, one can
define (as in the introduction) the exponential graph $C_k^H$.  The
vertices of $C_k^H$ are functions from $V(H)$ into $V(C_k)$ and two
such functions $f$ and $g$ are adjacent if for all edges $uv$ of $H$,
$f(u)$ and $g(v)$ are adjacent in $C_k$.  H\"aggkvist, Hell, Miller
and Neumann Lara proved that if there is no homomorphism from $H$ to
$C_k$, then $\CE$ has a homomorphism to
$C_k$~\cite{haggkvist1988multiplicative}.  Their proof can be viewed
as a generalization of the work of El-Zahar and Sauer, who proved the
same statement when $k=3$.  In fact, as in the case in the latter
proof of El-Zahar and Sauer, it is also implicit in the proof of
H\"aggkvist et al. that if we consider an (arbitrary) edge $ab$ in any
odd cycle of $H$, then functions $f$ in which $f(a) = f(b)$ form a
hitting set of the odd cycles in $\CE$.  It is therefore not
surprising that we can extend our framework for obtaining explicit
homomorphisms to odd cycles.

As in the case of $C_3$, we can find a homomorphism from $\CE$
to $C_k$ by considering an arbitrary odd cycle in $H$.  Applying Lemma
7 from \cite{haggkvist1988multiplicative}, we see that if there is no
homomorphism from $H$ to $C_k$, then $H$ contains an odd cycle with an
even number of fixed points.\footnote{In
  \cite{haggkvist1988multiplicative}, a fixed point is called a {\em
    2-point}.}  We refer to this odd cycle as $C_{2n+1}$.  It remains
to generalize the two key properties of adjacent functions (i.e.,
Lemmas \ref{label-lemma} and \ref{little-path-lemma}).
For an ordered pair of vertices $uv$ (i.e., an {\em arc} $uv$) with
values $c(u)$ and $c(v)$ (from $V(C_k)$), respectively, we say the
{\em value} of $uv$ is $\Delta(c(u)c(v))$, where
\begin{equation}\label{cycle-rules}
    \Delta(ij) = 
    \begin{cases}
      +1, & \text{if}\ j-i \equiv 2 \bmod k,\\
      -1, & \text{if}\ i-j \equiv 2 \bmod k,\\
      +\frac{1}{2}, & \text{if}\ j-i \equiv 1 \bmod k,\\
      -\frac{1}{2}, & \text{if}\ i-j \equiv 1 \bmod k,\\
       0, & \text{if}\ i-j \equiv 0 \bmod k.
    \end{cases}
  \end{equation}
For example, we have
\begin{eqnarray*}
\Delta(13) = \Delta(24) = \Delta(k2) = +1 ~\text{ and }~
\Delta(31) = \Delta(42) = \Delta(2k) = -1. 
\end{eqnarray*}
Let $f$ and $g$ be two adjacent functions from $V(C_{2n+1})$ to
$V(C_k)$.  We can apply the rules from \eqref{cycle-rules} to compute
the values $\ell_f, \ell_g$ and $p_f, p_g$ via Definitions
\ref{def:chord-cycle} and \ref{def:little-path}.  Note that if
$|f(u)-f(v)| \not \equiv \{0,2\} \bmod k$ for some arc $uv \in \vc$,
then $f$ is an isolated function in $C_k^{C_{2n+1}}$.  Moreover,
note that if $|f(u)-g(v)| \not \equiv 1 \bmod k$ for some arc $uv \in
C_{2n+1}$, then $f$ and $g$ are not adjacent.

The next lemmas are the generalizations of Lemma \ref{label-lemma} and
Lemma \ref{little-path-lemma} for homomorphisms to an odd cycle $C_k$.

\begin{lemma}\label{label-lemma-homomorphism}
Let $f$ and $g$ be two functions from $V(\CC)$ into $V(C_k)$ such
that $f$ and $g$ are adjacent.
Then $\ell_{f} = \ell_{g}$.
\end{lemma}

\begin{proof}
Recall the definition of the directed cycle $D$ and $\Delta(D)$ from
\eqref{directed-cycle} and \eqref{delta-cycle}.  It is
straightforward to prove that $\ell_f = \ell_g = \Delta(D)$. 
\end{proof}

\begin{lemma}\label{little-path-lemma-homomorphism}
Let $f$ and $g$ be two functions from $V(C_{2n+1})$ into $V(C_k)$
such that $f$ and $g$ are adjacent and let $\ell = \ell_f =\ell_g$.
Then $\ell -1 \leq p_f + p_g \leq \ell + 1$.
\end{lemma}

\begin{proof}
Let $a = u_1$ and
$b=u_{2n+1}$.   We observe that
\begin{eqnarray*}
p_f + p_g & = & 
\Delta(D) - \Delta(f(b)g(a)) -
  \Delta(g(b)f(a)) \\
& = & \ell - \Delta(f(b)g(a)) -
  \Delta(g(b)f(a)).
\end{eqnarray*}
Since
\begin{eqnarray*}
-1 & \leq & \Delta(f(b)g(a)) +
  \Delta(g(b)f(a)) \quad \leq \quad 1, 
\end{eqnarray*}
the lemma follows.
\end{proof}

It is straightforward to extend Claims \ref{clm:eitherOr} and
\ref{clm:anticorrelation} to this generalized setting and we obtain
the following theorem.

\begin{theorem}\label{thm:main-homomorphism}
Let $f$ and $g$ be two functions from $V(C_{2n+1})$ into $V(C_k)$ such
that $f$ and $g$ are adjacent.  Then {\sc Color-Vertex} assigns $f$
and $g$ to adjacent vertices in $V(C_k)$.
\end{theorem}

\section{Discussion: Explicit versus efficient colorings}

For a graph $H$ such that $\chi(H) > 3$, the question of finding an
{\em explicit} 3-coloring of $K_3^H$ is closely related to---but not
exactly the same as---the question of finding an {\em efficient}
3-coloring of $K_3^H$.  A connected component of $K_3^H$ is either (i)
isolated (i.e., a single vertex), (ii) bipartite, or (iii)
3-chromatic.  For a function $f$ from $V(H)$ into $\{1,2,3\}$, it can
be efficiently determined (in time polynomial in the size of $H$)
whether or not $f$ is isolated.  For any given connected component of
$K_3^H$ that is bipartite or 3-chromatic, there exists an odd cycle
that can be found efficiently (as discussed in Section
\ref{from-coloring}) and this odd cycle can be used to obtain an
explicit and efficient 3-coloring for this component.
In other words, for a given connected component, after a polynomial
amount of preprocessing time (i.e., time to find an odd cycle and to
fix an orientation of its chord cycle and an edge $ab$ to use as input
for the {\sc Vertex-Color} routine), we give an {\em explicit} reason
(i.e., certificate) for assigning a particular color to a function $f$
in the given connected component.  In particular, the value $\pab$ of
the little path from $a$ to $b$ is such a short certificate.  In terms
of efficiency, for any
subset $S$ of functions in $V(K_3^H)$ belonging to a fixed connected
component, the total time required to color the subgraph induced on
$S$ is $O(|S|\cdot|H|)$.

Moreover, for a function $f$ belonging to {\em any} 3-chromatic
component of $K^H_3$, we can actually use an arbitrary fixed cycle
$C_{2n+1}$ from $H$ for the {\sc Color-Vertex} routine.  However, the
fact that we can use the same cycle in $H$ for any such $f$ follows
from the main result of El-Zahar and Sauer; Proposition 3.4 in
\cite{el1985chromatic} states that for such an $f$, {\em all} odd
cycles in $H$ have an even number of fixed points.  Note that the
results we have presented here do not imply a proof of this
proposition.  Thus, while such a function $f$ can in fact be assigned
a color efficiently (i.e., in time $O(|H|)$), this time complexity
is not implied solely by the results we have presented here.

For $f$ belonging to an arbitrary bipartite component of $K^H_3$,
using the approach presented in this paper, we can assign a color to
$f$ in time $O(|B| \cdot |H|)$, where $B$ is the set of functions for
which we searched for a new odd cycle containing an even number of
fixed points.  Note that $|B|$ is upper bounded by the number of
functions also belonging to bipartite components previously colored by
the algorithm.  In other words, the algorithm is input-sensitive;
before invoking the {\sc Color-Vertex} routine on $f$, we need to
check all odd cycles used so far (in the order used) until we find one
with an even number of fixed points with respect to $f$.  We leave it
as an open problem to find an algorithm that assigns a color in time
$O(|H|)$ to a vertex $f$ from a bipartite component of $K_3^H$, so
that the resulting coloring is a proper 3-coloring or possibly even a
2-coloring.  Finally, we note that we do not know how to efficiently
determine if a function $f$ belongs to a bipartite component (i.e.,
whether or not it contains at least one odd cycle with an odd number
of fixed points).

\section{Acknowledgements}

We thank St\'ephan Thomass\'e for telling us about Hedetniemi's
Conjecture and for, even earlier, telling us about the results in
\cite{cereceda2007finding}.

\bibliography{hedet4arxiv}

\end{document}